\documentclass[smallextended]{svjour3}       
\smartqed  
\usepackage{graphicx}
\usepackage{csquotes}
\usepackage[table]{xcolor}
\usepackage[authoryear]{natbib}
\usepackage{amsmath}
\usepackage{amssymb}
\usepackage{mathptmx}
\usepackage{epigraph}
\begin{document}

\title{Perfect Prediction in Minkowski Spacetime: Perfectly Transparent Equilibrium for Dynamic Games with Imperfect Information}
\titlerunning{Perfect Prediction in Minkowski Spacetime}


\author{Ghislain Fourny}


\institute{G. Fourny \at
              ETH Z\"urich \\
              Department of Computer Science \\
              \email{ghislain.fourny@inf.ethz.ch}\\
}

\date{May 10, 2019, Updated May 22, 2019}

\maketitle

\begin{abstract}

The assumptions of necessary rationality and necessary knowledge of strategies, also known as perfect prediction, lead to at most one surviving outcome, immune to the knowledge that the players have of them. Solutions concepts implementing this approach have been defined on both dynamic games with perfect information and no ties, the Perfect Prediction Equilibrium, and strategic games with no ties, the Perfectly Transparent Equilibrium.

In this paper, we generalize the Perfectly Transparent Equilibrium to games in extensive form with imperfect information and no ties. Both the Perfect Prediction Equilibrium and the Perfectly Transparent Equilibrium for strategic games become special cases of this generalized equilibrium concept. The generalized equilibrium, if there are no ties in the payoffs, is at most unique, and is Pareto-optimal.

We also contribute a special-relativistic interpretation of a subclass of the games in extensive form with imperfect information as a directed acyclic graph of decisions made by any number of agents, each decision being located at a specific position in Minkowski spacetime, and the information sets and game structure being derived from the causal structure. Strategic games correspond to a setup with only spacelike-separated decisions, and dynamic games to one with only timelike-separated decisions.

The generalized Perfectly Transparent Equilibrium thus characterizes the outcome and payoffs reached in a general setup where decisions can be located in any generic positions in Minkowski spacetime, under necessary rationality and necessary knowledge of strategies. We also argue that this provides a directly usable mathematical framework for the design of extension theories of quantum physics with a weakened free choice assumption.


\end{abstract}

\section{Introduction}
\epigraph{How often have I said to you that when you have eliminated the impossible, whatever remains, however improbable, must be the truth?}{\textit{Sir Arthur Conan Doyle,\\Sherlock Holmes,\\The Sign of the Four (1890)}}

\subsection{Nashian game theory and strong free choice}

Most game-theoretical solution concepts, based on the work of John \citet{Nash1951}, assume that players make their decisions independently from each other.

This can be expressed in counterfactual terms, formalized by \citet{Lewis1973} by saying that an agent Mary picks some stragegy S, the other agents anticipated that Mary was going to pick strategy S, and that if Mary had, counterfactually, picked a different strategy T, the other agents' anticipations \emph{would still have been that Mary was going to pick strategy S}.

This assumption that a decision is independent is often called free will or free choice in literature. In this paper, we refer to it as a \emph{strong free choice} assumption. Formally, it can be formulated by saying that, if something is correlated to the decision, then it could potentially have been caused by it \citep{Renner2011}. In special-relativistic terms, it comes down to saying that anything correlated to the decision must be its future light cone.

Nashian solution concepts are thus characterized by unilateral deviations when optimizing utility, where each agent can adjust their strategy by assuming that the opponents' strategies are left unchanged.

\subsection{Non-Nashian game theory and perfect prediction}

An emerging, different line of research, non-Nashian game theory, drops this assumption in various ways.

There is a specific class of non-Nashian game theoretical results based on the opposite assumption, namely, that agents are perfectly predictable, and that the prediction of a decision is perfectly correlated with this decision. This can be expressed in counterfactual terms by saying that an agent Mary picks some stragegy S, the other agents anticipated that Mary was going to pick strategy S, and that if Mary had, counterfactually, picked a different strategy T, the other agents \emph{would have anticipated that Mary was going to pick strategy T}. This idea was introduced by \citet{Dupuy2000} and is called Perfect Prediction. This can be interpreted as a weaker form of free will in which the agents \emph{could have acted otherwise}, but not precluding predictability.

Contrary to intuition, assuming that agents are perfectly predictable does not lead to a trivial theory, but instead translates to a fix-point problem: the apparent conflict between making a free decision and being predictable no matter what one does translates into finding those outcomes that are immune against their anticipation. Knowing which outcome will be reached in advance, the players play towards this outcome. \emph{In spite of} knowing the outcome in advance, the players play towards it. Put more boldly, the very anticipation of a specific outcome by the players \emph{causes} them to play towards this outcome. \citet{Dupuy2000} conjectured that, under this assumption, interesting and desirable properties emerge, such as existence, uniqueness and Pareto-optimality of the outcome at hand.

This conjecture was found to be correct for games in extensive form, with perfect information with no ties. The corresponding equilibrium is called the Perfect Prediction Equilibrium \citep{Fourny2018}. For games in normal form, the equilibrium is at most unique and Pareto-optimal, however, it does not exist for all games \citep{Fourny2017}. It was called Perfectly Transparent Equilibrium.

In this paper, we extend the Perfectly Transparent Equilibrium concept to any game in extensive form with imperfect information. The above counterpart equilibria are thus special cases of this new, more general equilibrium concept.

The underlying Kripke semantics \citep{Kripke1963} formalizing the reasoning with possible worlds, explicit accessibility relations and counterfactual functions was introduced in \citep{Fourny2017b}. It also applies, with a few adaptations, to the generalized Perfectly Transparent Equilibrium.

\subsection{Related work}

Other non-Nashian results are found in literature, most of them with assumptions intermediate between Nashian Free Choice (independent decisions) and Perfect Prediction (decisions are always correctly anticipated). Joe Halpern dubbed this intermediate spectrum ``translucency''.

Solution concepts include:

\begin{itemize}
\item Superrationality in symmetric games \citep{Hofstadter1983}; in such games, superrational agents take into account that they are both superrational, that they will thus make the same decision and that the outcome must be on the diagonal of the  matrix. The equillibrium is then reached by taking the highest payoff on the diagonal. Deviations are not unilateral, but are directly correlated, which is the reason for optimizing on the diagonal. Superrationality coincides with the Perfectly Transparent Equilibrium on symmetric games in normal forms whenever the latter exists -- thus, the Perfectly Transparent Equilibrium can be seen as the generalization of superrational behavior on all games with imperfect information. 
\item Translucent players, minimax-rationalizability and individual rationality \citep{Halpern:2013aa}, \citep{Capraro2015}. For translucent players, some information on their decisions may partially leak, leading to potentially non-unilateral deviations. In this setting, a weaker assumption than necessary rationality is taken: common counterfactual knowledge of rationality (CCBR). The difference is that other players would have been rational in the case of a deviation, but excluding the deviating agent. Minimax-rationalizability is a superset of rationalizability that builds on more conservative assumptions: a strategy is minimax-dominated if its best outcome is worse than the worst outcome of another strategy, which takes all potential correlated opponent's moves into account in case of a deviation. Individual rationality, formally identical to the decade-old concept known in the context of the folk theorem and repeated games, is shown to be relevant in the translucent context, because it also makes sense when deviations are not unilateral. 
\item Second-Order Nash Equilibria \citep{Bilo2011}, which provide a superset of Nash equilibria for one-shot games by considering small improving sets leading to a Nash equilibrium.
\item Joint-selfish-rational equilibrium \citep{Shiffrin2009}, which provides a bargaining framework allowing two agents to improve on the Subgame Perfect Equilibrium on extensive form games with perfect information, with rounds of Pareto-improving iterations. The Joint-selfish-rational equilibrium differs from the Perfect Prediction Equilibrium in that eliminated outcomes are only temporary eliminated, and may be reconsidered when comparing utilities. In our framework, eliminated outcomes are permanently considered impossible in the rest of the reasoning.
\item Program equilibria \citep{Tennenholtz2004}, which provide a framework in which the agents can provide and see their own source code, which prevents any betrayal. Any individually rational outcome, in other words, an outcome that can be a steady state in a repeated game (folk theorem), can be implemented as a program equilibrium. In spite fo the transparency provided by source code access, the counterfactual implications in program equilibrium are essentially Nashian, i.e., there is only knowledge of strategies in the actual world.
\end{itemize}

\section{Games with imperfect information}

We start with the core definitions of a game in extensive form with imperfect information, as typically encountered in literature.

What  differentiates a game with imperfect information from a game with perfect information is that, when making a choice, there is some opacity regarding other agents' choices, even though these choices are not in the future. For example, these other decisions are being made in separate rooms, with no communication. Nodes are thus grouped into information sets, and a choice of action has to be taken not knowing at which node, within this information set, one is playing.

Games with imperfect information thus also subsume games in normal form: a game in normal form can be expressed as a game in extensive form with imperfect information\footnote{In literature, these games can also be referred as games with perfect information and simultaneous moves \citep{Rubinstein1994}.}.

\subsection{Formal definition}
\label{section-definition-game}

We take as the definition of a game with imperfect information the same as that of \citet{Jackson2019} at Stanford, making a few more properties explicit. An alternate definition based on sequences of actions is given by \citet{Rubinstein1994}.

\begin{definition}[Game with imperfect information]
A game with imperfect information is a tuple $(N, A, H, Z, \chi, \rho, \sigma, u, I)$ where
\begin{itemize}
\item $N$ is a set of players.
\item $A$ is a set of actions (common across all players).
\item $H$ is a set of choice (non-terminal) nodes.
\item $Z$ is a set of outcomes.
\item $\chi \in H \rightarrow \mathcal{P}(A)$ is the action function, assigning each choice node to the set of available actions at that node.
\item $\rho \in H \rightarrow N$ is the player function, assigning each choice node to a player.
\item $\sigma \in H \times A \rightarrow H \cup Z $ is the successor function, assigning each pair of choice node and action (available at that choice node) to a choice node or outcome.

There are a few constraints on $\sigma$ to enforce that the game is a tree. First, it is injective. Second, $\sigma(h, a)$ is only defined if $a\in \chi(h)$. Lastly, $\sigma$ must organize the choice nodes and outcomes in a single connected component: there can only be one root, as opposed to a forest\footnote{This was implicit in the original definition.}.

\item $u \in N \times Z \rightarrow \mathbb{R}$ is the utility function, assigning each player and outcome to a payoff. Since we are only interested in pure strategies, payoffs are ordinal, not cardinal, meaning that it only matters how they compare, but their absolute values do not. Literature thus also models utilities with an order relation, with no explicit payoff numbers. Using numbers improves readability and makes it easier to talk about examples.

\item $I$, the information partition, is an equivalence relation on $H$ that is compatible with the player function as well as with the action function. By convention the equivalence relation is expressed in terms of information sets, which are partitions $(I_{(i,j)})_{(i,j)\in N\times \mathbb{N}}$, one for each player and integer index. Formally, it fulfils for any $i$ and $j$ that $h \in I_{i,j} \implies \rho(h) = i$ as well as $h, h' \in I_{i,j} \implies \chi(h) = \chi(h')$. 
\end{itemize}
\end{definition}

\subsection{Notations}

Since the definition of the solution concept in this paper involves large formulas, we use a few notations that make them easier to read.

For payoffs, we write $u_i(z)$ rather than $u(i, z)$ to denote the payoff of player $i$ at outcome $z$ for convenience.

For navigation in the tree, we write $h\oplus a$ for $\sigma(h, a)$.

We also write $I_{i,j} \oplus a=\cup_{n\in I_{i,j}} \{ n \oplus a \} $ for an information set $I_{i,j}$ to denote all its successor nodes for a specific action. This is consistent with extending the application of a function to a set of inputs, as is common in algebra literature.

Finally, we use the letter $\delta$ to denote the set of all the descendants of a node in the tree induced by $\sigma$.

\begin{definition}[Descendant function]
Given a game with imperfect information $(N, A, $ $H, Z, \chi, \rho, \sigma, u, I)$, the descendant function maps any choice node or outcome to the set of all its descendants, i.e., in its transitive and reflexive closure via $\sigma$. For $h\in H \cup Z$,
$$
\delta(h) = \{ h' \in H \cup Z | h=h' \vee\exists k\in \mathbb{N}^*, \exists (a_i)_{i=1..k} \in A^k,
h' = h \oplus a_1 ... \oplus a_k\}
$$
\end{definition}


Whenever we have two nodes $h \neq h'$ such that $h' \in \delta(h)$, we introduce the notation $h'_h$ as the only action $a$ such that $h' \in \delta(h\oplus a)$. It always exists because if $h'$ is a descendant of $h$, then it is in one of its subtrees.

Following standard mathematical practice, we also use the notation $\delta(I_{i,j})$ where $I_{i,j}$ is an information set, to denote the set of all descendants of all nodes in this information set.

\subsection{Canonical form of a game with imperfect information and perfect recall}

\label{section-canonical}

Some games are allowed by the definition given in Section \ref{section-definition-game}, but contain subtrees that cannot be reached, under pure strategies by the definition of the game itself\footnote{These subtrees could still be reached in the presence of randomness, when considering behavioral strategies: a dice is thrown at each information set, and may lead to different actions for nodes in the same information set.}.

This happens when a node is a descendant of another node that is in the same information set: if $h'\in \delta(h\oplus a)$ for some $a\in A$, then the nodes in $\delta(h'\oplus b)$ for any $b\neq a$ are never reached for any pure strategies, i.e., for any assignment of information sets to choices of actions.

Such subtrees can simply be pruned out with no changes to the semantics of the game\footnote{If we only consider pure strategies.}. We call this the canonical form.

\begin{definition}[Canonical form of a game]
A game with imperfect information is a tuple $(N, A, H, Z, \chi, \rho, \sigma, u, I)$ is in canonical form if no information set contains two node that are different, and one is the descendant of the other: 

$$\forall I_{i,j}\in I, \forall n \in I_{i,j}, \delta(n) \cap I_{i,j} = \{ n \} $$

A game can be put in canonical form by pruning all nodes $n$ that are the strict descendants of another node $m$ in their information set, replacing such a node n by its subtree $n \oplus n_m$.



\end{definition}

A game with perfect recall always is in canonical form. This is because if we have a node and one of its descendants in the same information set, the agent does not know whether they already made this decision or not: we have imperfect recall. The converse is not true: a game in canonical form may have imperfect recall, i.e., it can be that an agent does not remember \emph{another} decision that they already made.

We assume in the remainder of this paper that all the considered games have been made canonical.

\section{Games with imperfect information and Minkowski spacetime}
\label{section-spacetime}
We now give an interpretation of a subcategory of games in extensive form with imperfect information and perfect recall in the context of special relativity. Specifically, we show that, if we consider agents making decisions across spacetime, we can model this setup by building a game in extensive form with imperfect information and perfect recall (and thus in canonical form).

\subsection{Minkowski spacetime}

We consider a Lorentzian manifold in which the metric tensor is constant, i.e., it is the same at all positions across space and time. In this setting, the distance between two points is independent of the observer, assuming that inertial timeframes are obtained from each other by Lorentz transformations. This is known as Minkowski spacetime \citep{Minkowski1908}.

We are thus looking at a vector space $\mathbb{R}^n$ endowed with a non-degenerate, symmetric bilinear form (an $n\times n$ matrix).  We assume a metric signature $(n-1,1)$, having in mind that, in practice, this is commonly $(3,1)$. The first $(n-1)$ coordinates are known as space, the last one as time.

Given two events (vectors) in spacetime (our vector space), we can calculate their distance in spacetime with the bilinear form. The sign of this distance allows us to classify pairs of unequal vectors into one of two cases:

\begin{itemize}
\item either they are \emph{timelike}-separated, meaning that any observer (inertial timeframe) would see these two events occur after one another, in always the same order. We can thus say that one of the two events \emph{precedes} the other, because its time coordinates are smaller than the other event's time coordinates in any inertial timeframe.
\item or they are \emph{spacelike}-separated, meaning that the order in which these events occur depends on the observer. No signal can be sent between these two events, because it would involve faster-than-light travel, equivalent to travelling back in time for some other observer.
\end{itemize}

\subsection{Decision points}

Now that we have a Minkowski spacetime, we can place in it what we can call decision points at which agents in a set $\hat{N}$ make decisions taken from a global set of actions $\hat{A}$.

We denote a decision point $\hat{I}_{i,j}$ in which agent $i\in \hat{N}$ makes the decision, and $j\in\mathbb{N}^*$ indexes all the decision points at which agent $i$ makes the decision. We call $\hat{I}$ the set of all decision points.

We denote $\hat{\chi}(\hat{I}_{i,j})\subseteq \hat{A}$ the set of possible actions that agent $i$ can make at decision point $\hat{I}_{i,j}$.

Figure \ref{figure-decision-points} shows an example with six decision points located in spacetime and four agents.

\begin{figure}
\centering\includegraphics[width=0.7\textwidth]{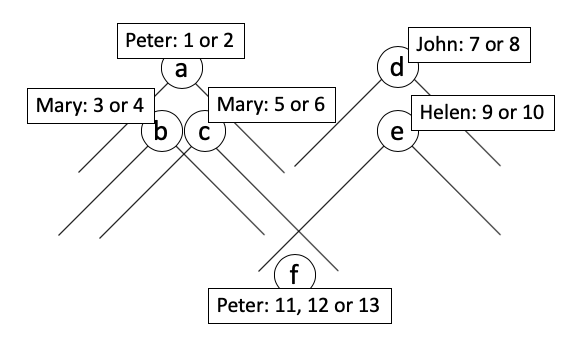}
\caption{Six decisions points, located in Minkowsky spacetime. Lines indicate their future light cones to make timelike and spacelike separation visible.}
\label{figure-decision-points}
\end{figure}

Each decision point has a location in spacetime. Given two decision points, we can thus say that they are either spacelike-separated, or timelike-separated based on their coordinates -- or that they coincide if the distance is 0. The precise locations are irrelevant to us. We use the partial order $\prec$ to denote timelike-separation, with the convention that two decision points at the same location are considered spacelike separated.

Figure \ref{figure-dag} shows, for the same example, the DAG making the timelike-separation partial order explicit.

\begin{figure}
\centering\includegraphics[width=0.5\textwidth]{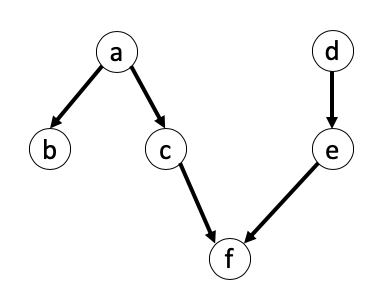}
\caption{The directed acyclic graph representing the timelike separation partial order $\prec$.}
\label{figure-dag}
\end{figure}

It is thus possible to list all decision points in a certain order that is compatible with their spacetime coordinates, i.e., whenever two decision points $I$ and $J$ are are such that $I\prec J$, then $I$ must precede $J$ in the list. We also require, to keep the list clean, that decision points that are co-located must also appear together in the list. Such a list is not unique, because spacelike-separation gives a few degrees of freedom, but it always exists. This is because, if it were impossible to build such a list because of a cycle in the order of events, we would be looking at a closed timelike curve, which does not exist in Minkowski spacetime\footnote{but may exist in general relativity}.

We denote this list of decision points $(\hat{I}_1, \hat{I}_2, ..., \hat{I}_n)$ where $n$ is the total number of decision points. Thus, when an decision point has one index, we mean its absolute position in the ordered list selected in the former paragraph. When it has two indices, we mean that the first index is the agent, and the second index its (arbitrary) index within the agent's decision points.

In our example, the list of decision points we take is $a, b, c, d, e, f$. Note that other lists compatible with timelike separation would be as acceptable: $a, d, b, c, e, f$ or $d, e, a, b, c, f$ for example.

It is crucial to distinguish between two orders: a partial order with timelike-separation semantics relative to spacetime, and a total order based on the selected ordered list of decision points, the latter being a superset of the former. We will denote the former $\prec$ and the latter $<$.

\subsection{Contingency coordinates}

Each decision point has, in addition to its spacetime coordinates, contingency coordinates. The contingency coordinates of a decision point are the actions that must be taken at previous, timelike-separated decision points for this decision point to actually be reached. In other words, the ability to make that decision must be caused\footnote{The notation of causality, in this paper, coincides with timelike-separation. $A$ causes $B$ if $A\prec B$. We carefully separate the notion of causality from that of counterfactuality.} by one specific set of events at decision points preceding it in spacetime.

For example, Peter at a first decision point $I_1$ may pick action $a$ or $b$, and John at a second decision point $I_2\nsucc I_1$, may pick $c$ or $d$. If Peter has picked $a$ and John $c$, then Mary at a third decision point $I_3$ such that $I_3\succ I_1$ and $I_3\succ I_2$ can pick $e$ or $f$. We say that the contingency coordinates of Mary's decision point are $(a, c)$, because Mary is only given this choice in case Peter picked $a$ and John picked $c$, and otherwise not.

The contingency coordinates of decision point $\hat{I_k}$ are thus an assignment of an action $a_{k,l}$ to decision points $\hat{I_l}<\hat{I_k}$ (so $a_{k,l}$ is only defined for $l < k$), but with two constraints:
\begin{itemize}
\item Only those decision points $\hat{I_l}$ such that $\hat{I_l}\prec\hat{I_k}$ are assigned an action. The others are assigned a dummy value $a_{k,l}=\perp$, which is only a notational convention.
\item A decision point $\hat{I_l}$ is only assigned an action $a_{k,l}$ if $\forall 1\le m \le l-1, (a_{l, m} \neq \perp \implies a_{l,m} = a_{k,m})$. Otherwise, $a_{k,l}=\perp$.
\item Conversely, if a decision point $\hat{I_l}\prec\hat{I_k}$, and $\forall 1\le m \le l-1, (a_{l, m} \neq \perp \implies a_{l,m} = a_{k,m})$, then $a_{k,l}$ must be defined, i.e., $a_{k,l}\neq \perp$ .
\end{itemize}

The contingency coordinates of all decision points can be represented on paper as a two-dimensional triangle of actions. Figure \ref{figure-triangle} shows an example of how contingency coordinates can look like in our running example. It can be seen that these coordinates are consistent. For example, the contingency coordinates of $f$ are $(2, \perp, 5, 7, 10)$, meaning that Peter makes a decision at $f$ only if he decided 2 at $a$, Mary decided 5 at $b$, John decided 7 at $d$ and Helen decided 10 at $e$.

Let us detail the row with the contingency coordinates of $f$. Column $a$ must have a non-dummy action, because $a, f$ fulfil the criterion (a universal quantifier on the empty set is always true). Column $b$ must have a dummy action, because $b, f$ do not fulfil the criterion: in column $a$ of decision point $b$, there is a 1, but there is a 2 in column $a$ of decision point $f$. Column $c$ must have a non-dummy action, because $c, f$ fulfil the criterion: the column $a$ of decision points $c$ and $f$ match, and column $b$ of decision point $c$ has a dummy action. Column $d$ must have a non-dummy action, because $d, f$ fulfil the criterion: all columns of decision point $d$ have a dummy action. Column $e$ must have a non-dummy action, because $e, f$ fulfil the criterion: the column $d$ of decision points $e$ and $f$ match, and column $a,b,c$ of decision point $e$ have a dummy action).

\begin{figure}
\centering\includegraphics[width=\textwidth]{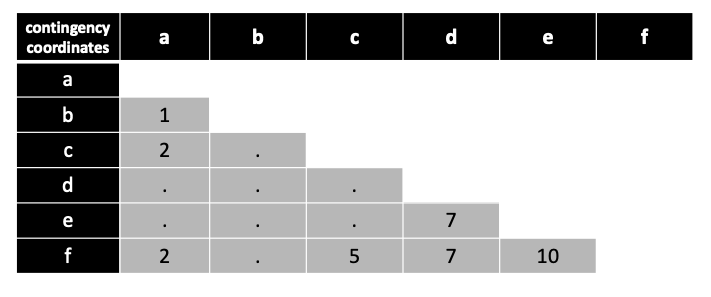}
\caption{The contingency coordinates of our six decision points. Note that the contingency coordinates are additional data to be supplied, i.e., the contingency coordinates must be provided as part of the setup and cannot be inferred from the partial order DAG or the spacetime positions.}
\label{figure-triangle}
\end{figure}

\subsection{Possible histories}

We now look at all possible worlds that can be instantiated from the decision points, i.e., considering all possible actions that can be taken.

A complete history is an assignment of actions to decision points, which we denote $(h_m)_{1\leq m\leq n}$, allowing for some decision points to be unassigned ($\perp$). This happens when a decision is not happening in the considered possible world, because the past does not cause making that decision.

We model the agent's preferences with a function $\hat{u}$ mapping each history and agent to a number with ordinal semantics, which is equivalent to a total order relation over histories for each agent.

An incomplete history is an assignment of actions to the first $k<n$ decision points, but not to the remaining ones. We use history for either complete or incomplete histories.

Like contingency coordinates, a history must respect constraints, i.e., when the past history of a decision point (list of actions taken at previous decision points in the total order) matches its contingency coordinates, this decision point must be assigned an action, and conversely.

Formally, an incomplete history $(h_1, ..., h_{l-1})$, up to decision point $I_l$, matches its contingency coordinates $(a_{l,1}, $ $ ..., a_{l, l-1})$ if $\forall 1 \leq m \leq l-1, (a_{l,m}\neq \perp \implies h_m=a_{l,m})$. Note that this is the exact same condition as the consistency between the contingency coordinates of two decision points, with the history acting as the latter decision point.

A history is consistent if, all its prefix histories $(h_1, ..., h_{m-1})$ with $1\le m \le n$ match the contingency coordinates of $\hat{I}_m$ if and only if $h_m\neq \perp$. It follows directly from this definition that, if a history is consistent, then all its prefixes are consistent as well.

We denote the set of all consistent complete histories $\hat{Z}$, and the set of all consistent incomplete histories $\hat{H}$, but taking only those $(h_1, ..., h_{m-1})$ that match the contingency coordinates of $I_m$, that is, such that a decision has to actually be made at the next step. \footnote{This is an easier way to express that we build an equivalence relation on consistent (complete or incomplete) histories by trimming $\perp$ assignments and comparing the remaining non-dummy actions, and that we quotient the set of consistent histories by this equivalent relation.}.

In our example, the consistent complete histories are

\begin{itemize}
\item $1, 3, \perp, 7, 9, \perp$
\item $1, 3, \perp, 7, 10, \perp$
\item $1, 3, \perp, 8, \perp, \perp$
\item $1, 4, \perp, 7, 9, \perp$
\item $1, 4, \perp, 7, 10, \perp$
\item $1, 4, \perp, 8, \perp, \perp$
\item $2, \perp, 5, 7, 9, \perp$
\item $2, \perp, 5, 7, 10, 11$
\item $2, \perp, 5, 7, 10, 12$
\item $2, \perp, 5, 7, 10, 13$
\item $2, \perp, 5, 8, \perp, \perp$
\item $2, \perp, 6, 7, 9, \perp$
\item $2, \perp, 6, 7, 10, \perp$
\item $2, \perp, 6, 8, \perp, \perp$
\end{itemize}

In our example, the consistent incomplete histories are

\begin{itemize}
\item $\emptyset$
\item $1$
\item $1, 3, \perp$
\item $1, 3, \perp, 7$
\item $1, 4, \perp$
\item $1, 4, \perp, 7$
\item $2, \perp$
\item $2, \perp, 5$
\item $2, \perp, 5, 7$
\item $2, \perp, 5, 7, 10$
\item $2, \perp, 6$
\item $2, \perp, 6, 7$
\end{itemize}

\subsection{Construction of the game with imperfect information}

The decision points can then be seen as a partition of $\hat{H}$, in which an incomplete history of size $m-1<n$ is mapped to information set $\hat{I}_m$.

We define the function $\hat{\rho}$ mapping an incomplete history $(h_1, ..., h_{m-1})$ to the agent in $\hat{N}$ who decides at decision point $\hat{I}_m$.

We define the function $\hat{\sigma}$ mapping an incomplete history, possibly empty if $m=1$, $(h_1, ..., h_{m-1})\in \hat{H}$ and an action $a\in \hat{\chi}(\hat{I}_m)$to the (incomplete or complete) history $(h_1, ..., h_{m-1}, a, \perp, ..., \perp)\in\hat{H}\cup\hat{Z}$, where we add as many $\perp$ actions as necessary for the resulting incomplete history, of size $l-1>m-1$, to match the contingency coordinates of $\hat{I}_l$, or all the way to a complete history if no such $l$ exists. The resulting history is thus consistent by construction. $\hat{\sigma}$ is injective by construction, because the previous incomplete history and action can be reconstructed straightforwardly. The unique root of the so obtained tree is the empty history.

For example, $\hat\sigma((1, 3, \perp), 8) = (1, 2, \perp, 8, \perp, \perp)$.

$(\hat{N}, \hat{A}, \hat{H}, \hat{Z}, \hat{\chi}, \hat{\rho}, \hat{\sigma}, \hat{u}, \hat{I})$ is then a game in extensive form with imperfect information, on which we can compute equilibria (Nash, PTE, etc). This game has a natural interpretation in which the players are agents located in Minkowski spacetime making decisions. The information sets are interpreted as the situations in which decisions are spacelike-separated and thus no signal can be sent between two agents.

If we require that the same agent cannot make a decision at two spacelike-separated (including colocated) decision points, then the game has perfect recall. This requirement is natural and corresponds to the fact that the timeline\footnote{the set of all spacetime coordinates ever occupied by an agent} of an agent must follow a timelike curve, that no observer in Minkowski spacetime can see the same agent at two spacelike-separated positions\footnote{This could, however, occur in general relativity in the presence of closed timelike curves, which is out of the scope of this paper.}, and that an agent can only make one decision at any time\footnote{We mean here his own time.}.

\begin{figure}
\centering\includegraphics[width=\textwidth]{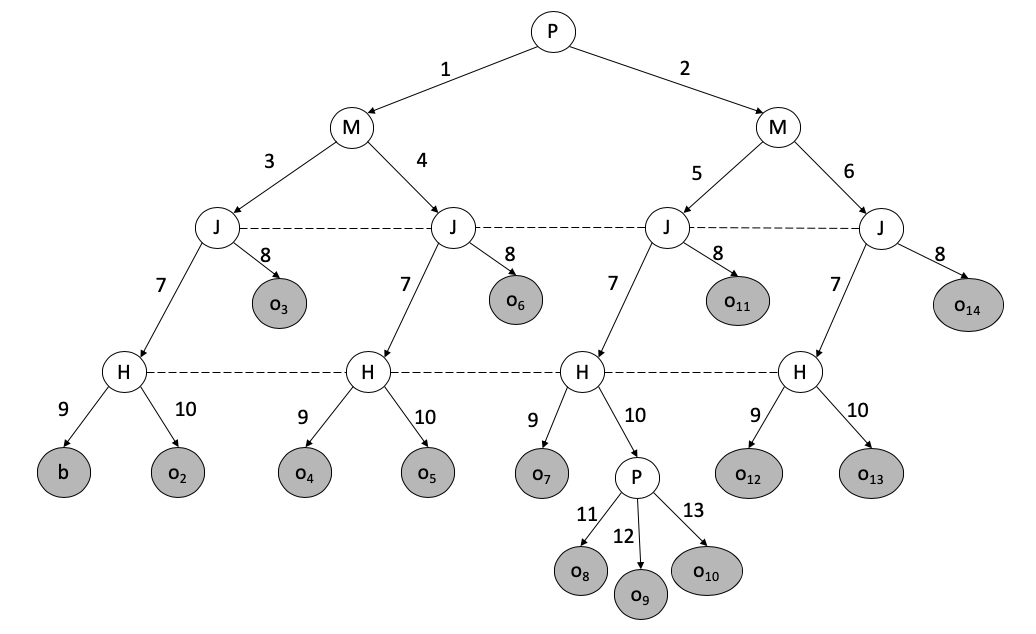}
\caption{The game associated with the example spacetime setup with six decision points and four agents. Each consistent complete history corresponds to an outcome (gray node) and each consistent incomplete history corresponds to a choice node (white node). The letters are the initials of the agents making decisions. The actions are indicated on the edges. Information sets are shown with dashed lines. The outcomes are numbered arbitrarily. The preferences of the agents are not shown, but would appear as tuples of four numbers at each outcome with ordinal semantics.}
\label{figure-game}
\end{figure}

Not all games with imperfect information and perfect recall can be obtained in this way. Further work includes characterizing this subclass of games for which an underlying semantics involving agents making decisions in Minkowski spacetime exists.

In the rest of this paper, we do have in mind that subclass of games as our prominent use cases, however, the generalized Perfectly Transparent Equilibrium is defined on all canonical games, even if they cannot be interpreted with agents located in Minkowski spacetime.

\section{Computation of the Perfectly Transparent Equilibrium}

We now consider any canonical game in extensive form with perfect information, and give the algorithm for computing the Perfectly Transparent Equilibrium. This game may or may not have been built from an underlying decision setup in Minkowski spacetime, and may or may not have perfect recall. However, we do assume that it is in canonical form.

The construction of the (potential) Perfectly Transparent Equilibrium is done by iterative elimination of those outcomes that cannot be the result of the game, assuming that the players perfectly predict them and are rational in all possible worlds.

We need to start with an important clarification. In the Nashian resolution of a game with imperfect information, a pure strategy corresponds to an assignment of an action to \emph{each} information set. In a Nashian setup, this makes sense because these assignments to information sets that are not actually reached model counterfactuals \citep{Lewis1973} under unilateral deviations: if other agents had picked a different pure strategy, these information sets would have been reached and the corresponding actions would have been chosen.

In a non-Nashian setup, such an assignment to all information sets makes no sense in reality, because given an outcome, decisions are only actually made at the information sets on the equilibrium path\footnote{In the spacetime setup, the semantic interpretation of this is that a decision is only made at a decision point if its contingency coordinates match the past history of decisions. A decision must be \emph{caused} by the conjunction of previous decisions in the sense that it lies in the future light cone of the past decisions' having been made in a specific way.}. Indeed, since we do not assume unilateral deviations, the counterfactual structure is fundamentally different. For example, the action chosen at a specific information set may be correlated to the other player's strategy rather than be independent.

A pure strategy in a setup with perfect prediction only maps chosen actions to those information sets that are on the equilibrium path leading to the equilibrium outcome. This map is completely determined by the outcome. A deviation from a decision is formally modelled with explicit counterfactual functions with a possible worlds semantics \citep{Fourny2017b}: the impact of a deviation to a different action, at an information set, is analyzed by looking at the closest world in which this alternate action is taken. This alternate possible world, with a potentially different outcome, may thus map actions to a different group of information sets, specifically, only those on the path leading to this alternate outcome.

As \citet{Dupuy2000} points out on centipede games, this avoids the backward induction paradox, because the reasoning leading to the equilibrium does not involve any reasoning at nodes outside the equilibrium path. In other words, the path leading to the Perfectly Transparent Equilibrium is causally consistent. A forward induction reasoning flags outcomes as logically impossible under our assumptions, while advancing on a path until only one final outcome remains. 

Figure \ref{figure-example-1} shows a game in extensive form played by two agents, Peter and Mary, with no ties and with imperfect information and, for the sake of illustration, imperfect recall. Perfect recall is not required in general. The computation does not require the absence of ties, but our theorems of uniqueness and Pareto-optimality require that there are no ties.

\begin{figure}
\centering\includegraphics[width=\textwidth]{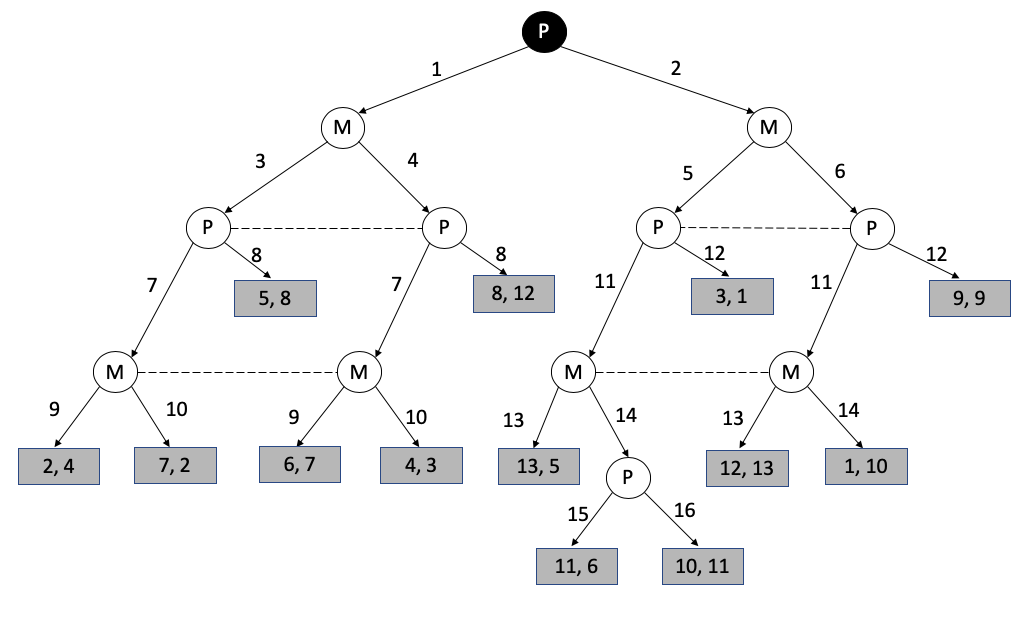}
\caption{The game used as an illustration of the computation of the PTE. This game has imperfect recall because Mary choosing between 13 and 14 cannot remember whether she previously picked 5 or 6. There are four information sets with more than one choice node. The letters P and M indicate whether Peter or Mary is making the decision. The payoffs are given as pairs: the first number is Peter's payoff, the second one is Mary's payoff. Black choice nodes mark reached information sets.}
\label{figure-example-1}
\end{figure}

We start with an initial set of outcomes containing all the outcomes of the game, $\mathcal{S}_0=Z$. For each $k$, we eliminate more outcomes from $\mathcal{S}_k$, building $\mathcal{S}_{k+1}$. In our example, $\mathcal{S}_0$ contains all of the thirteen outcomes. We will express subsequent sets of outcomes by graying out outcomes directly on the figure.

\subsection{Reached information sets}

Once it has been established that any outcome outside $\mathcal{S}_k$ is impossible, for some $k\ge 0$, it is possible to derive the information that some information sets must be reached by the game: no matter what outcome would be reached, the information set intersects with the path leading to it, meaning that reaching this outcome involves making a decision at this information set no matter what.

\begin{definition}[Reached information set]
Given a game with imperfect information  $(N, A, H, Z, \chi, \rho, \sigma, u, I)$ in canonical form, and given the set $\mathcal{S}_k$ of surviving outcomes at step $k\in\mathbb{N}$, an information set $I_{i,j}$ is reached at step $k$ if:

$$\mathcal{S}_k \subseteq \delta(I_{i,j})$$

We denote $\mathcal{R}_k$ the set of all such information sets at step $k$.

Reached information sets grow with each step, i.e., the algorithm is nothing else than a forward induction going from the root all the way to at most one surviving outcome.

In our example, we mark reached information sets in black. At step 0 (Figure \ref{figure-example-1}), only the information set containing the root (a singleton) is reached: all outcomes in $\mathcal{S}_0$ are in its descendance.

\subsection{Preemption}

Knowing that some information sets are guaranteed to be reached, we can eliminate outcomes that cannot possibly be known as the equilibrium if the agents are rational in all possible worlds.

Such an outcome $o\in\mathcal{S}_k$ is characterized by the fact that for some agent making a decision at some reached information set, a deviation from $o$ by picking a different action that the one leading to $o$ \emph{guarantees} to this agent a minimum payoff that is greater than the one she would obtain at $o$.

Preempted outcomes are characterized as those that do not Pareto-dominate the \emph{maximins} of the agents, i.e., the highest payoff they can guarantee themselves by a smart choice of actions. If an outcome, for some agent, is worst than its maximin, and the agent knows that this is the result of the game, then this is incompatible with their rationality: being rational, they would have picked a different action. This is a proof by \emph{reductio ad absurdum} that this outcome cannot be reached by the game. Maximin domination is very similar to the normal-form version of the Perfectly Transparent Equilibrium.

Preemption can only be done by deviating at information sets that are reached by the game. This is because otherwise, a preemptive action could itself be preempted, as originally pointed out by \citet{Dupuy2000} in his seeding work. If a preemption is carried out at an information set that is known to be reached, then no subsequent preemption can invalidate this preemption reasoning later on. This dependency resolution is addressed with the forward induction mechanism. This is the exact same idea as in the Perfect Prediction Equilibrium: the reasoning only involves nodes that are on the equilibrium path, which avoids backward induction paradoxes \citep{BIP}, because the justification of each choice is never motivated by what would happen at an information set that is actually not reached.

\end{definition}
\begin{definition}[Preempted outcomes]
Given a game with imperfect information $(N, A, $ $H, Z, \chi, \rho, \sigma, u, I)$ in canonical form, and given the set $\mathcal{S}_k$ of surviving outcomes at step $k\in\mathbb{N}$ as well as the set of reached information sets $\mathcal{R}_k$ at step $k$, we say that an outcome $z\in\mathcal{S}_k$ is preempted by player $i\in N$ at a reached information set $I_{i,j}\in\mathcal{R}_k$ if:

$$u_i(z) < \max_{
\begin{matrix}
a\in \chi(I_{i,j}) \\ 
\text{s.t.} \\
\mathcal{S}_k \cap \delta(I_{i,j} \oplus a)\neq \emptyset 
\end{matrix}
} \min_{
\begin{matrix}
z'\in \mathcal{S}_k \cap \delta(I_{i,j} \oplus a)
\end{matrix}
} u_i(z') $$

We call $\mathcal{P}_k$ the set of all outcomes that are preempted for some reached information set $I_{i,j}$ by some player $i\in N$.

\end{definition}

Note that, when a preemption of an outcome $z$ occurs, the argmax $a\in\chi(I_{i,j})$ is never $z_{I_{i,j}}$\footnote{This notation generalizes the notation $z_n$ where $n$ is a choice node. It is consistent because the game is in canonical form and for any outcome, only one choice of action in $I_{i,j}$ can cause $z$.} because $u_i(z)$ is never strictly smaller than the minimum in the subtree to which $z$ belongs, $z_{I_{i,j}}$.

At step 0 and for games in normal form, this definition coincides with that of individual rationality; it can thus be seen as a generalization of individual rationality to games with imperfect information.

\begin{figure}
\centering\includegraphics[width=\textwidth]{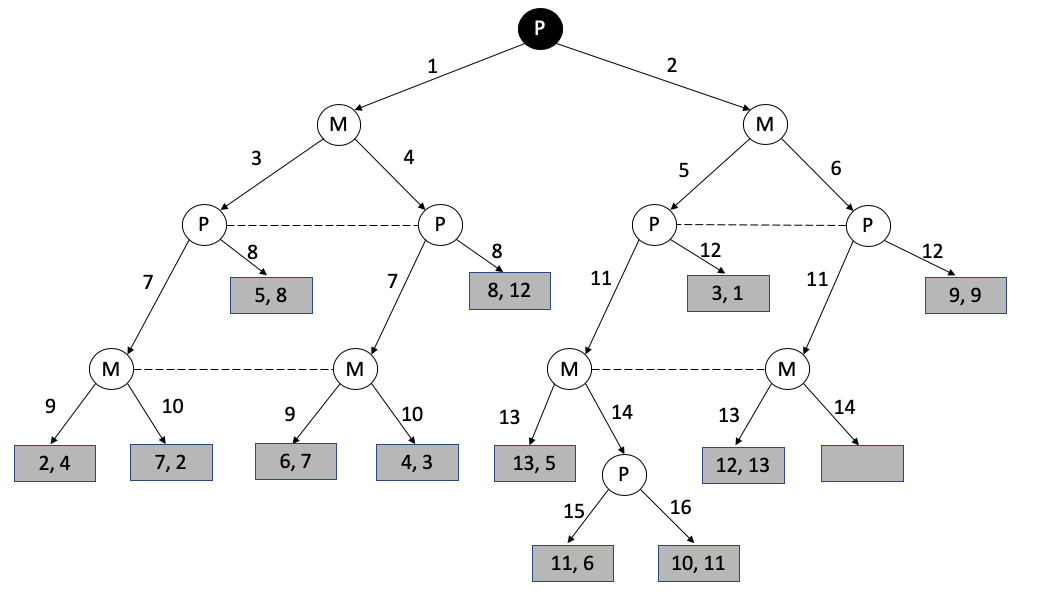}
\caption{At step 0, there is one outcome that is preempted: (1, 10), grayed out here.}
\label{figure-example-2}
\end{figure}

Figure \ref{figure-example-2} illustrates preemption on our example: outcome $(1, 10)$ is preempted, because Peter, if he chooses action 1 at the root information set, has a guaranteed payoff of at least 2 no matter what other decisions are made. His maximin is 2. Since $1<2$, $(1, 10)$ cannot be possibly reached under perfect prediction and rationality in all possible worlds. If Peter knew in advance that $(1, 10)$ was going to be the outcome, then he would deviate and pick action 1, leading to a causal inconsistency as $(1, 10)$ can only be caused by action 2. Again, the notion of causality in the perfect prediction framework is nothing else than a matter of future light cones, i.e., Grandfather's paradoxes must be avoided and the equilibrium path must be causally consistent.

\subsection{Recursion}

Knowing the preempted outcomes at step $k$, one deduces the new set of surviving outcomes, $\mathcal{S}_{k+1}$.

\begin{definition}[Surviving outcomes]
Given a game with imperfect information $(N, A, $ $H, Z, \chi, \rho, \sigma, u, I)$ in canonical form, and given the set $\mathcal{S}_k$ of surviving outcomes at step $k\in\mathbb{N}$ as well as the set of reached information sets $\mathcal{R}_k$ at step $k$ and the preempted outcomes $\mathcal{P}_k$ at set $k$, the surviving outcomes at step $k+1$ are:

$$\mathcal{S}_{k+1} = \mathcal{S}_k \setminus \mathcal{P}_k$$

The recursion is initialized with $\mathcal{S}_0=Z$.
\end{definition}

In our running example, the surviving outcomes are those not grayed out ($\mathcal{S}_1)$) on Figure \ref{figure-example-2}. In the next step, there is still only one reached information set at the root, and Peter's new maximin is 3. Thus, outcome $(2, 4)$ is preempted: knowing that $(1, 10)$ has previously been proven to be impossible, Peter would deviate to action 2, breaking the causal bridge. 

\begin{figure}
\centering\includegraphics[width=\textwidth]{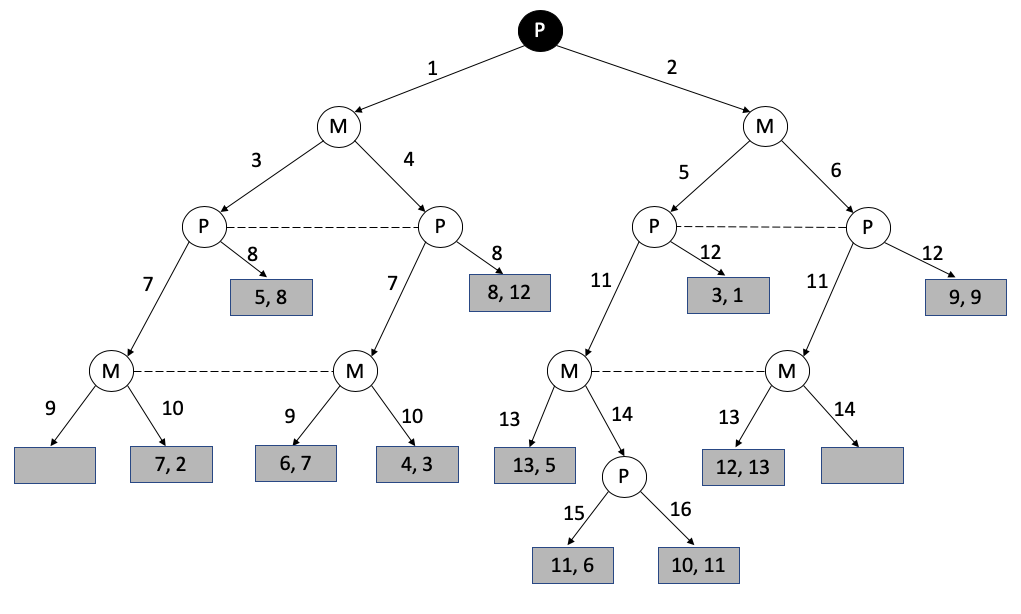}
\caption{This figure shows the outcomes at step 2. One outcome was preempted in step 1: (2, 4), grayed out here. The new maximin of Peter is 3.}
\label{figure-example-3}
\end{figure}
 
Figure \ref{figure-example-3} shows the set of outcomes $\mathcal{S}_2$. At step 2, Peter's maximin is 9. Outcome $(3, 1)$ is preempted.

\begin{figure}
\centering\includegraphics[width=\textwidth]{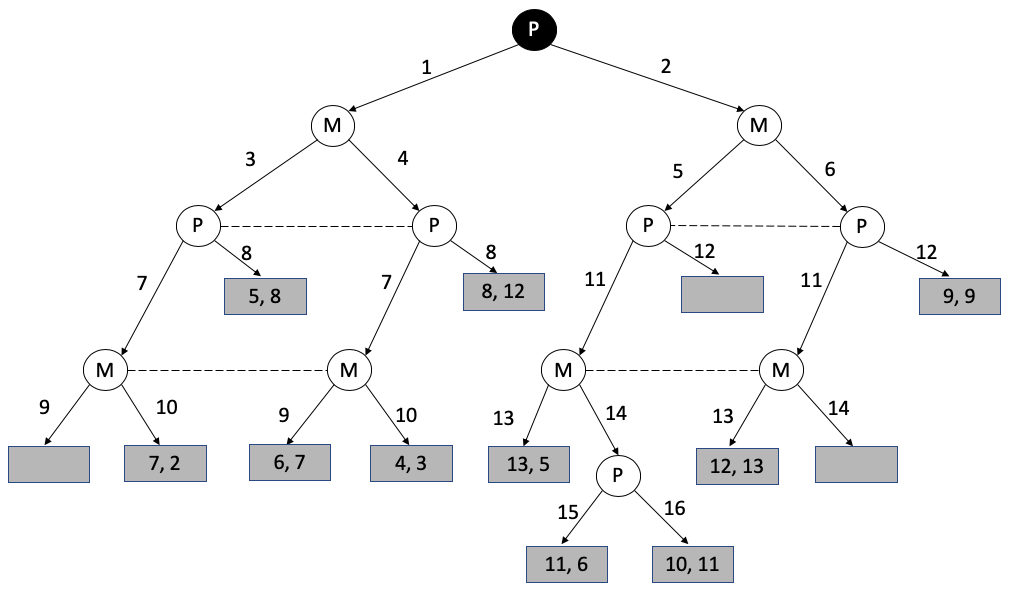}
\caption{This figure shows the outcomes at step 3. One outcome was preempted in step 2: (3, 1), grayed out here. The new maximin of Peter is 9.}
\label{figure-example-4}
\end{figure}

At step 3, shown on Figure \ref{figure-example-4}, the entire subtree below action 1 gets eliminated\footnote{Technically, the outcomes are eliminated, but it implies that the whole subtree cannot be reached either.}, because it is dominated by Peter's maximin of 9: Peter prefers action 2 no matter what.

Two more reached information sets appear at step 4 as shown on Figure \ref{figure-example-5}. The new maximin is (10, 9) for respectively Peter and Mary. Outcomes $(12, 5)$ and $(11, 6)$ are below Mary's maximin (9) and are preempted. Outcome $(9, 9)$ is below Peter's maximin (10) and is preempted. A new information set is reached.

Step 5 is shown on Figure \ref{figure-example-6}. The new maximin is (12, 13).

\begin{figure}
\centering\includegraphics[width=\textwidth]{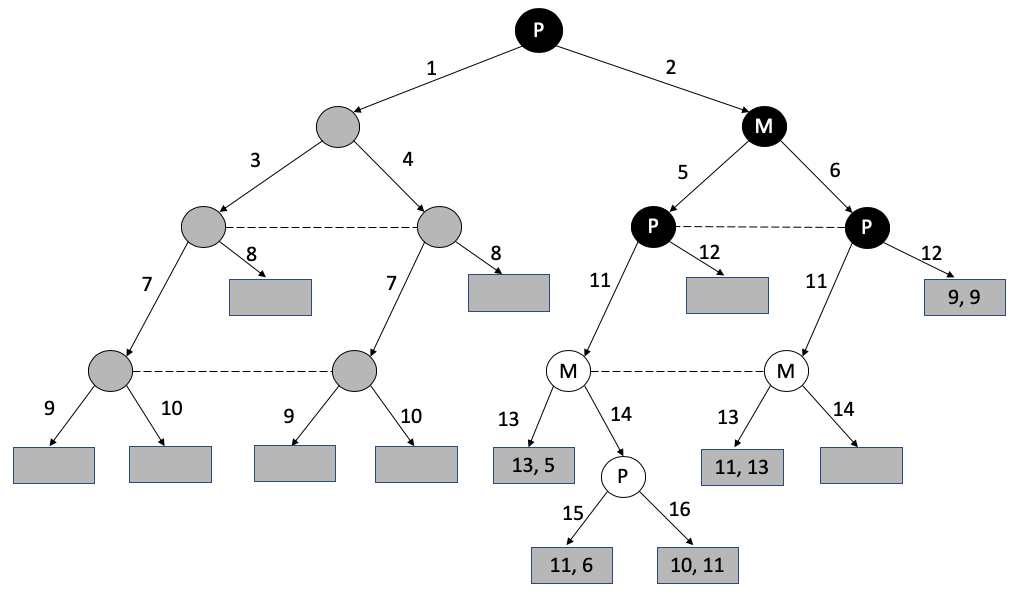}
\caption{This figure shows the outcomes at step 4. All outcomes in the left subtree were preempted in step 3, because they are all offering payoffs of less than 9 to Peter. The whole subtree is shown grayed out. We now have two more reached information sets: one with Mary picking 5 or 6, and one with Peter picking 11 or 12. Peter's maximin is 10 (guaranteed with action 11) and Mary's maximin is 9 (picking action 6)}
\label{figure-example-5}
\end{figure}

\begin{figure}
\centering\includegraphics[width=\textwidth]{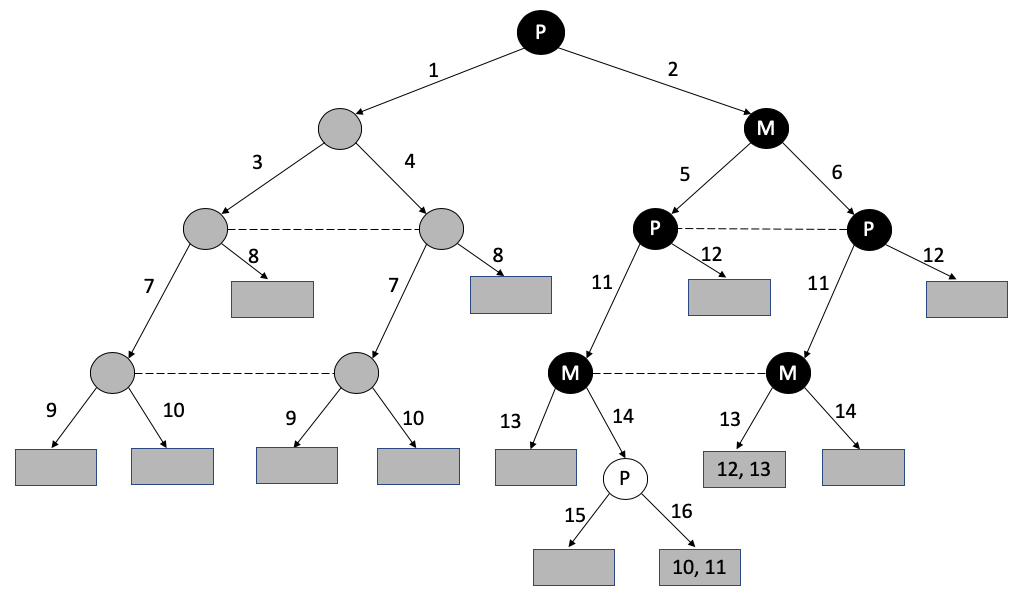}
\caption{This figure shows the outcomes at step 5. Three more outcomes have been eliminated.}
\label{figure-example-6}
\end{figure}

\clearpage

\subsection{The equilibrium}

The equilibrium then consists of the surviving outcomes.

\begin{definition}[Perfectly transparent equilibrium]
Given a game with imperfect information $(N, A, $ $H, Z, \chi, \rho, \sigma, u, I)$ in canonical form, an outcome $z\in Z$ is a perfectly transparent equilibrium if it survives all rounds of elimination:

$$z \in \displaystyle  \bigcap_{k\in \mathbb{N}} \mathcal{S}_k$$
\end{definition}

Figure \ref{figure-example-7} shows the final step with only one outcome left. No more preemption occurs, i.e., all subsequent steps stay the same. $(12, 13)$ is the Perfectly Transparent Equilibrium. This outcome is unique on this example. 

\begin{figure}
\centering\includegraphics[width=\textwidth]{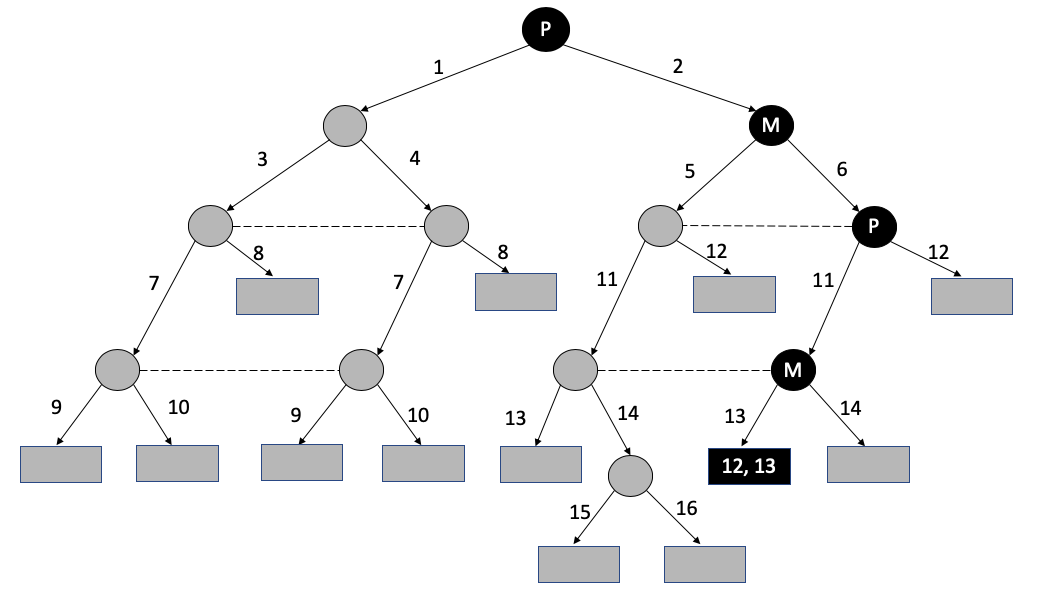}
\caption{The final, steady step with only one surviving outcome: the Perfectly Transparent Equilibrium.}
\label{figure-example-7}
\end{figure}

Actually, this surviving outcomes is always at most unique for games with no ties in the payoffs. -- but it could also not exist.

Note that the strategies of the two players are only assigning actions on the equilibrium path: Peter picks action 2, then Mary picks action 6, then Peter picks action 11, then Mary picks action 13. These decisions are causally consistent: they cause one another, and the anticipation that the final outcome is $(12, 13)$ also does not lead to any deviation. The decisions at any other information sets are undefined\footnote{However, the counterfactual structure is well defined, as explicitly given by the cascade of preemptions. Nowhere in this counterfactual structure is any decision at any other information set than those reached, involved}, as explained earlier.

\begin{theorem}[Uniqueness]
Given a game with imperfect information $(N, A, H, Z, \chi, $ $\rho, \sigma, u, I)$ in canonical form and with no ties, the perfectly transparent equilibrium is unique.
\end{theorem}

\begin{proof}[Uniqueness]
The sequence $(\mathcal{S}_k)_k$ is strictly decreasing until it reaches a singleton or the empty set. At every step $k$, if we consider the lowest common ancestor of all surviving outcomes $n$, the information set $I_{i,j}$ it belongs to is reached. Let us consider the player $i$ playing at this information set, and the outcome $z$ with the lowest payoff for this player (which is unique because there are no ties). 

Let us now consider an action $a$ different from $z_n$, so that $\delta(n\oplus a)$ intersects with $\mathcal{S}_k$. Such an action $a$ exists, because otherwise $n$ would not be the lowest common ancestor: $n \oplus z_n$ would be as well. Thus we have $\delta(I_{i,j}\oplus a)\cap \mathcal{S}_k \neq \emptyset$, meaning that the maximum at least involves another subtree.

Since the maximum also considers a different subtree at $n$, this maximum is greater than $u_i(z)$ because we picked this value to be minimal and there are no ties between payoffs for any agent.

Thus, the limit of the sequence $(\mathcal{S}_k)_k$ can only be a singleton or an empty set.\qed
\end{proof}

Finally, this surviving outcome, if it exists, is always Pareto optimal.

\begin{theorem}[Pareto optimality]
Given a game with imperfect information $(N, A, $ $H, Z, \chi, \rho, \sigma, u, I)$ in canonical form and with no ties, the perfectly transparent equilibrium, if it exists, is Pareto optimal in $Z$.
\end{theorem}

\begin{proof}[Pareto optimality]
Assuming the PTE $z$ is not Pareto optimal. Let $z'$ be a Pareto improvement of $z$. Let $k$ be the step at which $z'$ was eliminated, $i$ the player who preempted it and $I_{i,j}$ the reached information set at which it was preempted. Thus $u_i(z')$ is smaller than the agent's maximin, considering all players and all reached information sets. But since $z'$ is a Pareto improvement of $z$, we have $u_i(z)<u_i(z')$, so that z is preempted at this step too, which contradicts the fact that it is the PTE.\qed
\end{proof}

\subsection{Special cases}

As is known in literature, a game in normal form is a special case of games with imperfect information and perfect recall, in which each player has only one information set, and these information sets are organized below each other to cover the entire cartesian product of possible choices of actions. For these games, all information sets are reached already at step 0, which simplifies the computation down to what was contributed in \citep{Fourny2017}, i.e., an iterative elimination of outcomes that are not individually rational

A game in extensive form with perfect information is a special case of games with imperfect information and perfect recall, in which information sets are all singletons. In these games, the reasoning is simplified because there is only one non-trivial\footnote{i.e., where there is more than one choice of action left} reached information set at every step. In other words, at each step, only one player can preempt outcomes until all subtrees but one are eliminated. The elimination process comes down to \citep{Fourny2018} with the two principles of choice (rational choice, preemption).

\section{Link with quantum theory and current avenue of research}

This construct with decision points in spacetime is to put in direct perspective with \citet{Renner2011}'s spacetime random variables (SV) in thought experiments modelling (i) physicists picking a measurement axis and (ii) the choice (by the universe) of the measurement outcome. \citet{Renner2011} showed that the predictive power of quantum theory cannot be improved under a strong assumption of free choice.

The perfectly transparent equilibrium framework can directly be used to construct alternate theories of physics based on weakened free choice \citep{Fourny2019}, which eliminates impossible worlds \citep{Rantala1982}\citep{Kripke1965}, and under some conditions singles out at most one specific actual world.

Specifically, games with imperfect information compatible with an underlying spacetime semantics provide an alternate framework to discuss quantum measurements involving a game between agents and the universe. Such a framework implies that the universe has some utility or preference function, an idea that is not unusual in the history of science. For example, as pointed by Dupuy in e-mail discussions, the path followed by light can be expressed as a minimization problem amongst possible worlds. Many problems in mechanics can be expressed as the minimization or maximization of some quantity, e.g., the principle of least action. These formulations can be recast as a utility maximization problem.

Thus, given any setup in spacetime with physicists picking measurement axes and carrying out experiments, i.e., measuring an observable, and an outcome being picked (by nature), it is possible to build a game as shown in Section \ref{section-spacetime} where each outcome corresponds to a possible world in the sense of Everett's many-worlds interpretation \citep{Everett1973}. There is a degree of freedom in how to attribute ``nature utilities'' to the various possible global outcomes, as well as ``human utilities'' (the latter probably motivated by classical economics), which provides a formal entry point for designing any number of extension theories. For any such utility assignment, the PTE can be computed as-we -go with a forward induction, singling out one of the possible worlds as being our actual world, it being at most unique, in such a way that literature such as \citep{Renner2011} is not contradicted.

\section{Conclusion}

We have generalized the perfect prediction framework to any games with imperfect information. The underlying interpretation of the Perfectly Transparent Equilibrium is that, given agents endowed with weak free choice (``they could have acted otherwise'') and making decisions at any locations in Minkowski spacetime, it is the only possible world compatible with the fact that they are necessarily rational, and that perfectly predict each other's choices.

This interpretation shows that perfect prediction cohabits well with imperfect information: information is not only deducted from actual signals, constrained by the speed of light; under a weaker form of free choice, information can also be deducted by logical reasoning, eliminating logically impossible worlds.


\bibliographystyle{spbasic}      

\end{document}